\documentclass[a4paper,11pt]{article}
\pdfoutput=1 

\usepackage{jheparxiv} 

\usepackage[T1]{fontenc} 
\usepackage{xcolor}
\usepackage{slashed}
\usepackage{tensor}
\usepackage{amsmath}
\usepackage{amssymb}
\usepackage{mathrsfs}
\usepackage{graphicx}
\usepackage{stmaryrd}
\usepackage{twistor}
\usepackage{commath}
\usepackage[utf8x]{inputenc}
\newmuskip\pFqmuskip

\newcommand{\veps}{\varepsilon}

\newcommand{\qq}{\mathfrak{q}}

\newcommand*\pFq[6][8]{%
  \begingroup 
  \pFqmuskip=#1mu\relax
  \mathchardef\normalcomma=\mathcode`,
  \mathcode`\,=\string"8000
  \begingroup\lccode`\~=`\,
  \lowercase{\endgroup\let~}\pFqcomma
  {}_{#2}F_{#3}{\left[\genfrac..{0pt}{}{#4}{#5};#6\right]}%
  \endgroup
}
\newcommand{\pFqcomma}{{\normalcomma}\mskip\pFqmuskip}

\title{Moyal deformations, $W_{1+\infty}$ and celestial holography}
\author[a]{Wei Bu,}
\author[b]{Simon Heuveline}
\author[b]{\& David Skinner}

\affiliation[a]{School of Mathematics \& Maxwell Institute for Mathematical Sciences,\\
University of Edinburgh, EH9 3FD,\\
United Kingdom \vspace{0.1cm}}
\emailAdd{w.bu@sms.ed.ac.uk}

\affiliation[b]{Department of Applied Maths \& Theoretical Physics,\\
Wilberforce Road, Cambridge CB3 0WA,\\
United Kingdom \vspace{0.1cm}}
\emailAdd{sph48@cam.ac.uk}
\emailAdd{d.b.skinner@damtp.cam.ac.uk}

\abstract{We consider the Moyal deformation of self-dual gravity. In the conformal primary basis, holomorphic collinear limits of the amplitudes of this theory show that it enjoys a perturbatively exact symmetry algebra $LW_\wedge$ that generalises $Lw_\wedge$, the loop algebra of the wedge algebra of $w_{1+\infty}$, which appears in self-dual gravity.}

\begin{document} 
\maketitle
\flushbottom

\section{Introduction}
\label{sec:Introduction}

At the classical level, self-dual solutions of the vacuum Einstein equations have long been known to be closely associated to the infinite dimensional Lie algebra $w_{1+\infty}$~\cite{Penrose:1976js,Boyer:1985aj,Park:1989fz}. At its heart, the connection arises on the one hand through the realisation of $w_{1+\infty}$ as the space of diffeomorphisms of the plane that preserve a Poisson bracket and, on the other, the fact that any Ricci-flat self-dual manifold is hyperk{\"a}hler (in the Euclidean case) or pseudo-hyperk{\"a}hler (in (2,2) signature) and so possesses an $S^2$s worth of covariantly constant Poisson brackets. In fact, the presence of a whole $S^2$ family of Poisson brackets enhances $w_{1+\infty}$ to its loop algebra $Lw_{1+\infty}$. All these structures naturally arise from the twistor description of such self-dual vacuum spaces, where a certain holomorphic Poisson bracket forms part of the free data of Penrose's non-linear graviton construction~\cite{Penrose:1976js}.\\

Recently, the role of $w_{1+\infty}$ in perturbative gravitational scattering amplitudes has been explored, see {\it e.g.}~\cite{Strominger:2021lvk,Ball:2021tmb,Himwich:2021dau}.  Specifically, {\it conformally soft gravitons}, obtained by a Mellin transform of the usual plane wave momentum eigenstates, were shown to furnish a natural basis of generators of the loop algebra of the wedge subalgebra of $w_{1+\infty}$, with an important augmentation we discuss below. We call this algebra $Lw_\wedge$. The structure constants of $Lw_\wedge$ are then extracted by considering the behaviour of the corresponding Mellin transformed amplitudes in the holomorphic collinear limit $\la ij\ra\to0$. The same features can be seen in a perturbative expansion of the non-linear graviton, obtained via a twistor sigma model~\cite{Adamo:2021bej,Adamo:2021lrv}.\\

Part of the excitement surrounding these results comes from their implications for celestial holography -- the hope that quantum gravity in asymptotically flat space-times may be described by a dual theory living either at $\scri$ or on its space of generators, known as the celestial sphere (or, in (2,2) signature, the celestial torus)~\cite{Raclariu:2021zjz,McLoughlin:2022ljp,Pasterski:2021raf}. Celestial holography is far less developed than holography for  asymptotically AdS spaces, so generic information about possible celestial dual theories -- such that they possess $Lw_\wedge$ symmetry -- is especially valuable. Furthermore,  if the gravitons are sourced by operators in some celestial dual theory, taking the collinear limit should correspond to bringing these operators close together, so that the structure of this algebra might be expected to be visible in their OPE \cite{Himwich:2021dau,Pate:2019lpp,Adamo:2021zpw,Bu:2021avc}.

Indeed, $w$-algebras were initially developed in the context of 2d conformal theories~\cite{Bakas:1989mz,Bakas:1989xu} that, in addition to the stress tensor, have conserved currents of spin $>2$. In particular, $w_{1+\infty}$ itself is the symmetry algebra of a 2d CFT with an infinite tower of conserved currents with arbitrarily high spin $h\in\mathbb{N}$. More precisely, $w_{1+\infty}$ is the {\it classical} symmetry of these theories, seen by taking Poisson brackets (now on the phase space of the 2d theory) of the conserved currents. This algebra receives quantum corrections. Treating the 2d higher spin theory as a quantum theory, $w_{1+\infty}$ is deformed to a significantly more complicated algebra known as $W_{1+\infty}$~\cite{Bakas:1989mz,Pope:1989sr}. Unlike $w_{1+\infty}$, this algebra admits a central charge. In addition, while in $w_{1+\infty}$ the commutator between operators of spins $h$ and $h'$ is a single operator of spin $h+h'-2$, in $W_{1+\infty}$ a whole tower of lower spin operators appears, arising from multiple contractions in the quantum OPE.\\

These ideas make it natural to hope, as was put forward in~\cite{Strominger:2021lvk}, that 
in a full theory of quantum gravity in asymptotically flat space-times, amplitudes should transform under the loop algebra of $W_{1+\infty}$. Currently, it is not clear whether or not this hope can be borne out. On the one hand, it was shown in~\cite{Ball:2021tmb} that in self-dual gravity, the simpler algebra $w_{1+\infty}$ receives no quantum corrections at any order in perturbation theory. More precisely, in any theory whose action consists of a Lagrange multiplier enforcing some version of the self-dual vacuum Einstein equations, for generic external momenta the only possible amplitudes are a 3-pt  $\overline{\rm MHV}$ tree amplitude, and an $n$-particle all $+$ amplitude that is exact at one loop. In both these amplitudes, \cite{Ball:2021tmb} found that it is $w_{1+\infty}$ rather than $W_{1+\infty}$ that appears. On the other hand, in the context of Yang-Mills, it was shown in~\cite{Costello:2022wso} that going beyond the self-dual sector meant that the known structure of the collinear limit of amplitudes was incompatible with the associativity expected of the OPE in any dual theory, with the obstruction already seen in 1-loop corrections to the collinear splitting function. While this associativity could be restored by adding in additional fields (in~\cite{Costello:2022wso}, a certain type of axion), these additional fields simultaneously cancelled the 1-loop all $+$ amplitude.\\

In this paper, we find a deformed version of self-dual gravity that does exhibit a quantum corrected $W$-algebra. However, in the bulk this deformation arises by replacing the Poisson bracket of self-dual gravity by a Moyal bracket, so is suggestive of non-commutative, rather than quantum, corrections.\\

Our paper is arranged as follows. In section~\ref{sec:algebra} we review the basic properties of $w_{1+\infty}$, $W_{1+\infty}$, their loop algebras and their wedge subalgebras. Our presentation closely follows the excellent review of~\cite{Pope:1991ig}, to which the reader is referred for a more comprehensive treatment. In section~\ref{Action_section}, after briefly reviewing the Chalmers-Siegel form~\cite{Chalmers:1996rq} of the action for self-dual gravity, we present its Moyal deformation, which will be our main object of study. Like self-dual gravity, the Moyal deformed theory is 1-loop exact. We evaluate its $\overline{\rm MHV}$ 3-pt tree amplitude in section~\ref{sec:Moyal-gravity}. In section~\ref{sec:OPEs} we study the collinear splitting function which is also 1-loop exact in the Moyal deformed theory, finding that the form of the splitting function is deformed in a simple way. By taking the Mellin transform of this collinear limit, we find that the amplitude is compatible with a dual 2d theory whose OPEs yield the loop algebra of the wedge subalgebra of a $W$-algebra.  We conclude in section~\ref{sec:Discussion} with a speculative conjecture about the possible form of the all plus one-loop amplitudes in this theory.\\

\emph{Note added:} While this paper was being prepared, \cite{Monteiro:2022lwm} appeared on the arXiv.  This paper has some overlap with the present work.

\section{Review of $w_{1+\infty}$ and its family of $W$-algebras}
\label{sec:algebra}

In this section we briefly review $w_{1+\infty}$ and its relation to Poisson diffeomorphisms of the plane. We then review the family of $W$-algebras which can be viewed as quantum deformations of $w_{1+\infty}$. We refer the reader to ~\cite{Pope:1989sr,Pope:1991ig} for a comprehensive review of these topics and their relations to higher spin symmetries of 2d CFTs. 

\bigskip

The infinite dimensional Lie algebra $w_{1+\infty}$ is spanned by generators $w^p_m$ for $p,m\in\Z$, with Lie bracket 
\begin{equation}\label{w_algebra}
    \left[w^p_m,w^q_n\right] = 2(m(q-1)-n(p-1))\,w^{p+q-2}_{m+n}\,.
\end{equation}
The element $w^1_0$ is central, and its presence distinguishes $w_{1+\infty}$ from $w_\infty$. The factor of 2 on the right of~\eqref{w_algebra} can be changed by a universal rescaling of the generators; we will find it convenient to use the convention given in~\eqref{w_algebra}.

This algebra arises in many different contexts. Of particular relevance to this paper is the fact~\cite{Bakas:1989mz,Bakas:1989xu,Hoppe:1988gk} that it can be represented as the space of diffeomorphisms of the plane that preserve a Poisson bracket. Letting $(u,v)$ be coordinates on this plane, the Poisson bracket is
\begin{equation}
\label{Poisson}
\left\{f,g\right\} = \frac{\partial f}{\partial u}\frac{\partial g}{\partial v} - \frac{\partial f}{\partial v}\frac{\partial g}{\partial u}\,,
\end{equation}
for any pair of (smooth) functions $f,g$. Diffeomorphisms that preserve this Poisson structure are generated by vector fields that are Hamiltonian, so $V = \{h, \ \}$ for some function $h$. The elements $w^p_m$ then form a basis for these Hamiltonians. In particular, one recovers~\eqref{w_algebra} by taking
\begin{equation}
\label{w_generator}
    w^p_m = u^{p+m-1}\,v^{p-m-1}\,,
\end{equation}
and using the Poisson bracket~\eqref{Poisson} as Lie bracket. In fact, in our context, the plane will naturally be the {\it complex} 2-plane $\C^2$, so  $(u,v)\in\C^2$ will be holomorphic coordinates and~\eqref{Poisson} a holomorphic Poisson bracket.

\medskip

The generators $\ell_m \equiv 2w^2_m$ of~\eqref{w_generator} with $p=2$ fixed, form a subalgebra of $w_{1+\infty}$ that we recognise as the Witt algebra
\begin{equation}
\label{Wittalgebra}
\left[\ell_m,\ell_n\right] = (m-n)\,\ell_{m+n}\,,
\end{equation}
It is well known that the Witt algebra is the classical limit of the Virasoro algebra
\begin{equation}
\label{Virasoro}
\left[L_m,L_n\right] = (m-n)\,L_{m+n} + cm(m^2-1)\,\delta_{m+n,0}\,.
\end{equation}
that characterises a 2d CFT of central charge $c$. The Virasoro generators $L_m$ are the Laurent modes of the (holomorphic) stress tensor $T(z)$ of the CFT, with the algebra~\eqref{Virasoro} appearing from the $TT$ OPE. The appearance of the central charge $c$ is a quantum effect in the CFT. From this perspective, it is natural to expect that $w_{1+\infty}$ also admits a quantization. This quantization was discovered~\cite{Fateev:1987zh,FATEEV1987644,Zamolodchikov:1985wn} by studying 2d CFTs with higher spin symmetry, and is known as $W_{1+\infty}$. Their generators $W^p_m$ are the Laurent modes of higher spin conserved currents. There is a large literature on $W_{1+\infty}$ (and other $W_N$-algebras) and their CFT realisations, and we refer the reader to {\it e.g.}~\cite{Pope:1991ig} for a review. 

In fact, $W_{1+\infty}$ is best understood as a member of a 1-parameter family of infinite dimensional Lie algebras that is sometimes called $W(\mu)$. These algebras each have generators $W^p_m$ for $p,m\in\Z$ and are defined by the relations\footnote{Note that we have shifted the upper index $p$ and $q$ of the generators by 2 compared to the labels in~\cite{Pope:1991ig}.} 
\begin{equation}\label{capitalW_algebra}
    \left[W_m^p,W_n^q\right]=\sum_{l\geq 0}{\qq}^{2l} f^{pq}_{2l}(m,n;s)\, W_{m+n}^{p+q-2l-2} + c_p(m)\,\qq^{2(p-2)}\,\delta^{p,q}\,\delta_{m+n,0}\,.
\end{equation}
Here, $\qq$ and $s$ are parameters, with $s$ related to $\mu$ by $\mu=s(s+1)$. The functions $f_{2l}^{pq}(m,n;s) = -f_{2l}^{qp}(n,m;s)$ are structure constants (depending on the parameter $s$) while $c_p(m)$ are central charges. By demanding that the bracket in~\eqref{capitalW_algebra} obeys a Jacobi identity, Pope {\it et al.} showed~\cite{Pope:1989sr} that the structure constants must take the form
\begin{subequations}
\begin{equation}
\label{W_coefficients}
    f^{pq}_{2l}(m,n;s)=\frac{1}{2(2l+1)!}\,\phi^{pq}_{2l}(s)\,N^{pq}_{2l}(m,n)\,,
\end{equation}
where
\begin{equation}
\label{hypergeometric}
    \phi_{2l}^{pq}(s)= \pFq[4]{4}{3}{-\frac{1}{2}-2s,\frac{3}{2}+2s,-l-\frac{1}{2},-l}{\frac{3}{2}-p,\frac{3}{2}-q,p+q-\frac{3}{2}-2l}{\, 1\,}
\end{equation}    
in terms of the generalized hypergeometric function ${}_4F_3$, and where
\begin{equation}\label{Symplecton_coefficient}
  N^{pq}_{2l}(m,n)= \sum_{i= 0}^{2l+1}(-1)^i\binom{2l\!+\!1}{i}[p\!-\!1\!+\!m]_{2l+1-i}\,[p\!-\!1\!-\!m]_{i}\,[q\!-\!1\!-\!n]_{2l+1-i}\,[q\!-1\!+n]_{i}
 \end{equation}
\end{subequations}
in terms of the descending Pochhammer symbol, defined by $[a]_b=a! / (a-b)!$ whenever $a\geq -1$ and $b\geq0$ . Similarly, the central charges are constrained to be
\begin{equation}
 c_p(m) = c\,\frac{2^{2p-7}\,p!\,(p-2)!}{(2p-3)!!\,(2p-1)!!}\, \prod_{k=1-p}^{p-1}(m-k)\,.
\label{W_centre}
\end{equation}
In particular, all central charges are fixed in terms of the Virasoro central charge $c$.

\medskip 

Let us make some remarks. Firstly, the parameter $\qq$ controls the deformation away from $w_{1+\infty}$, in the sense that~\eqref{capitalW_algebra} reduces to~\eqref{w_algebra} when $\qq\to0$. However, if $\qq\neq0$ it can be removed from~\eqref{capitalW_algebra} by rescaling $W^p_m \to \qq^{p-2}\,W^p_m$, 
so that the actual value of $\qq$ has no meaning. Secondly, because the hypergeometric function in~\eqref{hypergeometric} is invariant under $s\rightarrow -s-1$, the algebras are more properly labelled by $\mu=s(s+1)$. In particular, the algebra usually called $W_{1+\infty}$ corresponds to setting $\mu=-\frac{1}{4}$ (and so $s=-\frac{1}{2}$),  while $W(0)$ is usually called $W_\infty$. We note that all $W(\mu)$ algebras contain the Virasoro algebra as the subalgebra generated by $L_m = 2W^2_m$. Finally, notice that we have not imposed an upper bound of the value of $l$ in the summation in~\eqref{capitalW_algebra}. In fact, it turns out that the structure constants $f^{pq}_{2l}(m,n;s)$ are non-zero only when $p+q-2l-2\geq0$, so that for finite $p,q$ this summation  includes only finitely many terms. In particular, this ensures that, just as in $w_{1+\infty}$, the element $W^1_0$ remains central. 

\medskip

Unlike $w_{1+\infty}$ with its relation to Poisson diffeomorphisms of the plane, to the best of our knowledge, there is currently no known geometric realisation of $W(\mu)$ for generic $\mu$ (see {\it e.g}~\cite{Pope:1991ig} for a discussion). In particular, there is no known geometric realisation of $W_{1+\infty}$. However, 
there is a particular member of the $W(\mu)$ family, occurring  when $\mu=-\frac{3}{16}$ (and $c=0$), for which such an interpretation is known. Fixing $\mu=-\frac{3}{16}$ implies either $2s+\frac{1}{2}=0$ or $2s+\frac{3}{2}=0$, so that one of the arguments in the top line of the hypergeometric function in~\eqref{hypergeometric} vanishes. In either of these cases, the hypergeometric function reduces to 1, so that the structure constants in this algebra simplify and the relations~\eqref{capitalW_algebra} become
\begin{equation}
\label{W_algebra}
    \left[W_m^p,W_n^q\right]=\sum_{l\geq 0}\frac{{\qq}^{2l}}{2(2l+1)!} \,N^{pq}_{2l}(m,n)\, W_{m+n}^{p+q-2l-2} 
\end{equation}
when the central charge $c=0$.

The algebra~\eqref{W_algebra} can now be realised geometrically by equipping the $(u,v)$ plane with a Moyal bracket, deforming the earlier Poisson bracket~\cite{Pope:1989sr}. That is, we define the Moyal bracket of a pair of functions $f,g$ by~\cite{moyal_1949}
\begin{subequations} 
\begin{equation}
\label{Moyal_bracket}
    \{f,g\}_{\qq} = \qq^{-1}(f\star g - g\star f)\,,
\end{equation}
where the Moyal star product is given by
\begin{equation}
\label{Moyal_star}
    f\star g = f\, \exp\left[ \qq \left(\overleftarrow{\partial_u}\,\overrightarrow{\partial_v} - \overleftarrow{\partial_v}\,\overrightarrow{\partial_u}\right)\right]\,g\,.
\end{equation}
\end{subequations}
The Moyal bracket is a deformation of the Poisson bracket, in the sense that 
\begin{equation}
    \{f,g\}_{\qq} = \{f,g\} + \mathcal{O}(\qq)\ .
\end{equation} 
In fact, it is the unique deformation constructed purely from the Poisson bracket, such that the deformed bracket still obeys a Jacobi identity~\cite{cmp/1103922592,Fletcher:1990ib}. Just as with $w_{1+\infty}$, the algebra~\eqref{W_algebra} can be realised by acting with the Moyal bracket~\eqref{Moyal_bracket} on the generators $W^p_m = u^{p+m-1}\,v^{p-m-1}$. Deforming the Poisson bracket to the Moyal bracket thus corresponds to deforming $w_{1+\infty}$ to the $W(\mu)$ algebra at $\mu=-\frac{3}{16}$. For this reason, $W(-\frac{3}{16})$ is sometimes called the {\it symplecton algebra}. 

\medskip

Now we come to an important point. As mentioned in the introduction and as we review in section~\ref{sec:Chalmers-Siegel}, the algebra that is relevant to the context of celestial holography is not $w_{1+\infty}$ but rather an algebra we shall denote by $Lw_\wedge$. This is the loop algebra of 
an algebra $w_\wedge$ whose generators $w^p_m$ obey the same relations as those of $w_{1+\infty}$, but with two important differences in the allowed values of the indices. Firstly, the index $m$ is restricted to lie in the wedge region $1-p\leq m \leq p-1$, from which the algebra gets its name. Note that we require $p\geq1$ for this region to be non-empty. Secondly, we now allow $p$ and (simultaneously) $m$ to take \emph{half-integer values as well as integer ones.} The generators of $w_\wedge$ may be represented by  monomials $w^p_m = u^{p+m-1}\,v^{p-m-1}$ so that $w_{\wedge}$ generates Poisson diffeomorphisms that are holomorphic over $\C^2$. Note that for fixed $p$, there are $2p+1$ possible values of $m$, so that there are an odd number of generators $w^p_m$ when $p$ is an integer, but an even number where $p$ is a (odd) half-integer.

The loop algebra $Lw_\wedge$ arises because the generators $w^p_m(z)$ of $w_{\wedge}$ that appear in self-dual gravity depend on an additional parameter $z$. The generators of $Lw_\wedge$ are
obtained by the contour integrals 
\begin{equation}
\label{w_loop_generators}
    w^p_{m,r} =\oint \frac{dz}{2\pi i} \frac{w^p_m(z)}{z^{r}} \qquad\qquad\text{so that}\qquad\qquad 
    w^p_m(z)=\sum_{r\in\mathbb{Z}} w_{m,r}^p z^{r-1}\,.
\end{equation}
It follows that these generators obey relations
\begin{equation}
\label{Lw_algebra}
    \{w^p_{m,r},w^q_{n,s}\} = 2(m(q-1)-n(p-1))\,w^{p+q-2}_{m+n,r+s}\,,
\end{equation}
again with the wedge restrictions $|m|\leq p-1$ and $|n|\leq q-1$.\\

Now let us consider whether this extends to the $W(\mu)$-algebras.  Each of these $W(\mu)$-algebras possesses a subalgebra obtained restricting $|m|\leq p-1$, because the Pochhammer symbols in $N^{pq}_{2l}(m,n)$ ensure that the structure constants vanish if $|m+n|\leq p-q-2l-3$. However, remarkably the only case in which is is possible to augment these algebras by allowing half-integer values of $p$ is the symplecton algebra\footnote{We thank Roland Bittleston for clarifying this point to us.}. This is because, for generic $s$, the hypergeometric function in~\eqref{hypergeometric} diverges if any of the arguments on its bottom line is a negative integer. The exception is for $s(s+1)=-\frac{3}{16}$, {\it i.e.} the symplecton algebra, where the hypergeometric function reduces to 1. Thus, since including all the conformally soft gravitons requires the augmentation to half-integer $p,m$, the only possible extension of $Lw_\wedge$ as a Lie algebra is $LW_\wedge$. This may be viewed as the loop algebra of the wedge subalgebra of the symplecton algebra, with the same augmentation to include half-integer values of $p,m$.

\section{Self-dual gravity and its Moyal deformation}
\label{Action_section}

In this section we review actions for self-dual gravity on space-time. We show how the conformally soft modes of the (positive helicity) graviton correspond to generators $w^p_m(z)$ of $Lw_\wedge$, with $2p\in\mathbb{Z}$. We then consider a Moyal deformation of this theory and compute its 3-pt tree-level $\overline{\rm MHV}$ amplitude.

\subsection{The Chalmers-Siegel action for self-dual gravity}
\label{sec:Chalmers-Siegel}

In the absence of a cosmological constant, self-dual gravity can be described by the Chalmers-Siegel action~\cite{Chalmers:1996rq}
\begin{equation}
\label{Chalmers-Siegel}
    S[\tilde\phi,\phi] = \int \tilde\phi\left(\Box \phi + \frac{\kappa}{2}\left\{\partial^{\dal}\phi,\partial_{\dal}\phi\right\}\right)\,d^4x\,.
\end{equation}
Here $\phi$ and $\tilde\phi$ are scalar fields representing the positive and negative helicity states of the graviton, respectively, while $\kappa = \sqrt{32\pi G_{\rm N}}$ is the coupling. To write the interaction, we have defined $\partial_{\dal}= \alpha^\alpha(\partial/\partial x^{\alpha\dal})$ for some choice of spinor $|\alpha\rangle$, and also introduced 
\begin{equation}
\label{R4Poisson}
    \left\{f,g\right\} = (\partial^{\dal}\!f)\,(\partial_{\dal} g) 
    =\epsilon^{\dal\dot{\beta}}\,(\partial_{\dot{\beta}}f)\,(\partial_{\dal}g)
\end{equation}
as a Poisson bracket on $\R^4$. The presence of this Poisson bracket is the origin of the fact that amplitudes in self-dual gravity possess $w_{1+\infty}$ symmetry.  Notice that $\{\partial^{\dal}\phi,\partial_{\dal}\phi\}= \partial^{\dal}\partial^{\dot{\beta}}\phi\,\partial_{\dal}\partial_{\dot{\beta}}\phi$ so that the interaction involves four derivatives in total. 

This action may be understood as follows (see also {\it e.g.}~\cite{Siegel:1992wd,Adamo:2021bej}. Any self-dual Ricci-flat  $M$ is hyperk{\"a}hler\footnote{Or pseudo-hyperk{\"a}hler in (2,2) signature.} and so possesses an $S^2$ family of complex structures, labelled by the spinor $|\lambda\rangle$ up to scale.  A hyperk{\"a}hler manifold also has an $S^2$'s worth of symplectic structures, which for our 4-manifold $M$ are given up to scale by 
\begin{equation}
    \Sigma(\lambda) = \lambda^\alpha\nabla^{\dal}_{\ \alpha}\lrcorner\,(\lambda^\beta \nabla_{\dot{\beta}\beta}\lrcorner\,({\rm vol}(M)))= \frac{1}{2}e^{\dal\beta}\wedge e_{\dot{\beta}}^{\ \alpha}\,\lambda_\al\lambda_\beta\,.
\end{equation} 
Here $\nabla_{\dal\al}$ is the connection on the tangent bundle, ${\rm vol}(M) = \frac{1}{4!}e^{\dal\al}\wedge e^{\ \beta}_{\dal}\wedge e^{\dot{\beta}}_{\ \al}\wedge e_{\dot{\beta}\beta}$ is the volume form on $M$ and $e^{\dal\al}$ the vierbein 1-forms dual to $\nabla_{\dal\al}$. 

The hyperk{\"a}hler condition is equivalent to the triple $\Sigma^{\al\beta}=\e^{\dal(\al}\wedge e_{\dal}^{\ \beta)}$ of 2-forms being closed. In particular, we can identify an open patch $U\subset M$ with a patch of $\C^2$, by picking a basis $(|\alpha\ra,|\hat\al\ra)$ for our spinors (it will be convenient to choose $\la\al\hat\al\ra=1$) and letting $(u,v) = (x|\hat\alpha\rangle)^{\dal}$ be holomorphic coordinates in the complex structure defined by $|\lambda\rangle=|\hat\alpha\rangle$. In these coordinates, the vierbeins can be chosen to have components
\begin{equation}
\label{vierbeins}
 e^{\dal\al}\hat{\al}_\al = \d x^{\dal\al}\,\hat{\al}_{\al}\qquad\text{and}\qquad
 e^{\dal\al}{\al}_\al = \d x^{\dal\al}\al_{\al} - 
 \kappa\,\p^{\dal}\p_{\dot{\beta}}\phi\ \d x^{\dot\beta \beta}\hat\al_{\beta}
\end{equation}
for some scalar $\phi(x)$. The constant $\kappa$ controls the deformation away from flat space. With these vierbeins, closure of $\Sigma^{\al\beta}\hat{\alpha}_\al\hat{\alpha}_\beta$ and $\Sigma^{\al\beta}\hat{\al}_\al\al_{\beta}$ are automatic, while closure of the remaining $\Sigma(\al)$ requires that $\phi$ obeys
\begin{equation}
\label{2nd_Plebanski}
\Box\phi +\frac{\kappa}{2}\left\{\p^{\dal}\phi,\p_{\dal}\phi\right\}=0\,.
\end{equation}
This is known as the second Plebanski equation, and arises as the field equation by varying $\tilde\phi$ in of~\eqref{Chalmers-Siegel}.
Notice that the Poisson bracket in~\eqref{Chalmers-Siegel} \&~\eqref{2nd_Plebanski} is the inverse of the symplectic form $\Sigma(\hat\alpha)$ that has type $(2,0)$ in our chosen complex structure. 

\medskip

The choice of $|\alpha\rangle$ means the action~\eqref{Chalmers-Siegel} respects only a subgroup\footnote{The subgroup is $SU(2)\times B$ where $B$ is the Borel subgroup of $SU(2)$ represented by unimodular upper triangular matrices; {\it i.e.} the subgroup of $SU(2)$ that preserves the spinor $|\alpha\rangle$ up to scale.} of $SO(4)$. However, provided the external states of momentum\footnote{In Euclidean signature, the linearised on-shell condition $p^2=0$ requires that the external momenta are complex.} $p=|\lambda\ra[\tilde{\lambda}|$ are normalized as 
\begin{equation}
    \label{plane_wave_normalisation}
    \phi_p(x) = \langle\alpha\lambda\rangle^{-4}\,e^{{\rm i}p\cdot x}
     \qquad\qquad
    \tilde\phi_p(x) = \langle \alpha\lambda\rangle^4\,e^{{\rm i}p\cdot x}\,,
\end{equation}
the amplitudes it gives rise to are invariant under the full $SO(4)$, as we would expect from its origin as self-dual gravity. In fact, the only non-vanishing amplitudes of self-dual gravity are the tree-level amplitude with one negative helicity and $n-1$ positive helicity gravitons, and the $n$-particle, all $+$ amplitude at 1-loop. The tree-level amplitude vanishes unless all external particles are (holomorphically) collinear, but the 1-loop amplitude exists for generic $p_i$, subject only to $p_i^2=0$ and $\sum p_i=0$. Furthermore, these tree-level $-+\cdots +$ and 1-loop $+\cdots+$ amplitudes computed from~\eqref{Chalmers-Siegel} agree, at the same order of perturbation theory, with the corresponding amplitudes computed from the full Einstein-Hilbert action. However, in full gravity these amplitudes receive further loop corrections.

\medskip

The generators of the loop algebra $Lw_{\wedge}$ are usually described as coming from scattering conformally soft gravitons, rather than plane waves~\cite{Guevara:2021abz}. For the positive helicity outgoing graviton, these are obtained by parametrizing 
\begin{equation}
    \label{spinor_parametrize}
    \lambda_\alpha= \sqrt{\omega}\,(1,z) = \sqrt{\omega}\,z_\alpha\,,
    \qquad\qquad
    \tilde\lambda_{\dal} = \sqrt{\omega}\,(1,\tilde z) = \sqrt{\omega}\,\tilde{z}_{\dal} 
\end{equation}
and then taking the residue of the Mellin transformed momentum eigenstate
\begin{equation}
\label{Mellin}
G^\Delta_{z,\tilde{z}}(x) = \int_0^\infty \frac{d\omega}{\omega}\,\omega^{\Delta}\,\phi_p(x) = \frac{\im^{\Delta-2}}{\la\alpha z\ra^4}\frac{\Gamma(\Delta-2)}{(x^{\alpha\dal}z_\alpha \tilde{z}_{\dal})^{\Delta-2}}
\end{equation}
at integer values of $\Delta$. The normalisation factors in~\eqref{plane_wave_normalisation}, which ensure that $\phi$ and $\tilde\phi$ represent states of helicity +2 and $-2$ respectively, mean that the residue is non-zero only for $\Delta = k\in\{ 2,1,0,-1,-2,\ldots\}$.  The residues have conformal weights $(\frac{k+2}{2},\frac{k-2}{2})$, and in particular admit a (binomial) mode expansion
\begin{equation}
\label{tilde_z_mode_expansion}
    {\rm Res}_{\Delta=k}\,\left(G^\Delta_{z,\tilde{z}}(x)\right)= \frac{(-\im)^{2-k}}{\la\alpha z\ra^4}\frac{(x^{\alpha\dal}z_\alpha\tilde{z}_{\dal})^{2-k}}{(2-k)!}= \sum_{m=1-p}^{p-1} \frac{\tilde{z}^{p-m-1}\,w^p_{m}(z)}{(p-m-1)!\,(p+m-1)!}
\end{equation}
in $\tilde{z}$. Following~\cite{Adamo:2021lrv}, in the final equality we have relabelled $k= 4-2p$ to agree with the conventions in~\eqref{w_algebra}, and defined
the conformally soft modes
\begin{equation}
\label{conformal_modes}    
w^p_m(z) = \frac{(-1)^{p-1}}{\la\alpha z\ra^{4}}\, (x^{\alpha\dot{0}}z_\alpha)^{p+m-1}(x^{\alpha\dot{1}}z_\alpha)^{p-m-1}\,.
\end{equation}
These $w^p_m(z)$ are the generators of $Lw_\wedge$, with modes $w^p_{m,r}$ coming from further expanding in $z$.  The fact that the residues of $G^\Delta_{z,\tilde{z}}(x)$ involved only positive powers of $x^{\alpha\dal}z_\alpha \tilde{z}_{\dal}$ is the origin of the restriction to the wedge subalgebra. Note again that the indices $p,m$ can each be (simultaneously) either an integer or half-integer. The structure of the algebra itself will come from the interactions between these modes, and can be seen in the corresponding amplitudes. 

\medskip

Although it will not be central to this paper, we also point that, at the classical level, self-dual gravity and its relation to $Lw_\wedge$ has long been known to be closely related to twistor theory. See {\it e.g.}~\cite{Penrose:1976js} for an original reference and \cite{Adamo:2021lrv} for a modern treatment with emphasis on scattering amplitudes. As in self-dual Yang-Mills~\cite{Costello:2021bah,Costello:2022upu,Costello:2022wso}, the situation at the quantum level is more subtle, see~\cite{Bittleston:2022nfr}.

\subsection{Moyal deformed self-dual gravity}
\label{sec:Moyal-gravity}

The origin of $w_{1+\infty}$ in self-dual gravity amplitudes is ultimately the presence of the Poisson bracket~\eqref{R4Poisson} on $\R^4$. The fact that the symplecton algebra arises as the Lie algebra of functions on the plane under the Moyal bracket strongly suggests that, to obtain a theory whose amplitudes respect the deformed loop algebra $LW_\wedge$, we should deform the action by changing~\eqref{R4Poisson} to a Moyal bracket. That is, we consider Moyal deformed self-dual gravity, by which we mean the theory with action
\begin{equation}
    \label{Moyal_action}
    S_{\qq}[\tilde\phi,\phi]= \int \tilde\phi\left(\Box\phi + \frac{\kappa}{2}\left\{\partial^{\dal}\phi,\partial_{\dal}\phi\right\}_{\qq}\right)\,d^4x \,,
\end{equation}
where $\{\ , \ \}_{\qq}$ is the Moyal bracket defined via
$\{f,g\}_{\qq} = \qq^{-1}\left( f\star g - g\star f\right)$ with 
\begin{equation}
    \label{R4_Moyal_star}
    f\star g = f\, \exp\left[ \qq \left(\epsilon^{\dal\dot{\beta}}\,\overleftarrow{\partial_{\dal}}\,\overrightarrow{\partial_{\dot{\beta}}}\right)\right]\,g\,.
\end{equation}
the Moyal star on $\R^4\cong \C^2$, just as in~\eqref{Moyal_bracket}-\eqref{Moyal_star}. We emphasise that this Moyal bracket acts on space-time itself, rather than on phase space as is common in applications to deformation quantization~\cite{moyal_1949,Kontsevich:1997vb,Groenewold1946OnTP,Dito:1990rj,Fairlie:1998rf}. Moyal deformed self-dual gravity has been considered previously in~\cite{Strachan:1992em,Strachan_1995}, where the non-commutative version of the Plebanski equations and associated non-linear graviton construction were considered from the perspective of integrable systems. 
We note that this Moyal star product breaks Lorentz invariance to $SU(2)\times B$ where $B\subset SU(2)$ is a Borel subgroup that fixes the spinor $|\alpha\ra$. To introduce a Moyal bracket in a way compatible with full $SU(2)\times SU(2)$ invariance requires moving to a higher spin theory, see {\it e.g.}~\cite{Krasnov:2021nsq,Tran:2022tft}.

\medskip

Since the kinetic term in~\eqref{Moyal_action} is undeformed, at the linearised level we scatter the same states as in the Chalmers-Siegel theory. In particular, momentum eigenstates are again normalised as in~\eqref{plane_wave_normalisation}, so the Moyal deformed theory possesses the same set of conformally soft gravitons as usual.  This corresponds to the fact that, as we saw in section~\ref{sec:algebra}, $W_\wedge$ has the same set of generators as $w_\wedge$, with only the structure of the algebra itself being deformed. This deformation of course arises from the deformed interaction. While the Moyal bracket is a complicated, non-local operator on space-time, its action on momentum eigenstates is remarkably simple. We have
\begin{equation}
    \label{Moyal-momentum}
    \left\{\phi_{p_1},\phi_{p_2}\right\}_{\qq} = \frac{1}{\qq}\sinh\left(\qq \,\langle \alpha|p_1p_2|\alpha\rangle\right)\,\phi_{p_1}\phi_{p_2}
\end{equation}
for any pair of 4-momenta $p_{1,2}$ that may be off-shell.

In particular, the 3-particle $\overline{\rm MHV}$ tree amplitude that follows from~\eqref{Moyal_action} is given by $M^{0,3}_\qq= \delta^4(\sum_i p_i)\,\cM^{0,3}_\qq$, with
\begin{equation}
    \label{Moyal_MHV_bar}
\begin{aligned}
    \cM^{0,3}_{\qq}(p_1^-,p_2^+,p_3^+) &= \kappa\left(\frac{\langle\alpha 1\rangle}{\langle\alpha 2\rangle\langle\alpha 3\rangle}\right)^4\frac{[23]\langle\alpha 2\rangle\langle\alpha3\rangle}{\qq}\,\sinh\left(\qq\,[23]\langle\alpha2\rangle\langle\alpha3\rangle\right)\\
     &= \kappa\, \frac{[23]^7}{([12][23][31])^2}\,[23]_{\qq}\ .
\end{aligned}
\end{equation}
The second expression here holds on the support of momentum conservation and involves the deformed symplectic product of pairs of dotted spinors, defined as
\begin{equation}
\label{Moyal_deformed_square}
    [ij]_{\qq}\,=\, \frac{\sinh\left(\qq\,[ij]\la\alpha i\ra\la\alpha j\ra\right)}{\qq\,\la\alpha i\ra\la\alpha j\ra} \,.
\end{equation}
Notice that, just like the usual spinor product $[ij]=\epsilon^{\dal\dot\beta}\tilde\lambda_{i\dot\beta}\tilde\lambda_{j\dal}$, this deformed product obeys $[ij]_\qq = -[ji]_{\qq}$, behaves as $[ij]_\qq \mapsto (r_ir_j)^{-1} [ij]_\qq$ under the scaling $(\lambda_i,\tilde\lambda_i)\mapsto (r_i\lambda_i,r_i^{-1}\tilde\lambda_i)$, and is 
invariant under $SL(2)$ transformations acting on dotted spinor indices. However, the fact that $[ij]_\qq$ depends on a choice of undotted spinor $|\alpha\ra$ shows that amplitudes in the Moyal deformed theory are not fully Lorentz invariant. Finally, we notice that $\lim_{\qq\to 0}\, [ij]_\qq = [ij]$, so that the Moyal deformed amplitude~\eqref{Moyal_MHV_bar} reduces to the usual three-point $\overline{\rm MHV}$ tree amplitude as $\qq\to0$.

\section{Celestial OPEs and $LW_\wedge$}
\label{sec:OPEs}

In this section, we will see explicitly that the $Lw_{\wedge}$ symmetry of self-dual gravity gets deformed to $LW_\wedge$ in the Moyal theory. This $LW_\wedge$ symmetry is perturbatively exact.

\subsection{The holomorphic collinear limit of $\cM^{1,n}_\qq$}
\label{sec:Collinear}

In the celestial holography program,  collinear limits of amplitudes are interpreted as providing information about the structure of OPEs in a 2d theory living on the celestial sphere. In particular, in self-dual gravity these limits reveal that any such celestial dual theory must contain operators that generate $Lw_\wedge$. 

While Moyal theory~\eqref{Moyal_action} will also generate 1-loop all plus amplitudes, at present we do know understand their explicit form (though we tentatively present a possible form in section~\ref{sec:Discussion}. Fortunately, in any quantum field theory, the behaviour of amplitudes in the \emph{true} collinear limit $p_i\propto p_j$ is fixed on general grounds, with the $\ell$-loop, $n$ particle amplitude at  factorizing into a sum of $k\leq\ell$-loop, $n-1$ particle amplitudes and an $(\ell-k)$-loop splitting function that describes the how the two collinear particles connect to the remainder of the amplitude~\cite{Bern:1998sv,Bern:1998xc}. In both self-dual gravity and the Moyal theory, there only non-trivial amplitudes (for $n>3$) are the 1-loop all plus $\cM^{1,n}_\qq$, so only the tree-level splitting function is relevant. In particular, we must have
\begin{equation}
    \cM^{1,n}_\qq \xrightarrow{1\parallel 2} \text{Split}_{\qq}(1^+,2^+)\,\cM^{1,n-1}_\qq
\end{equation}
in the true collinear limit, where $\text{Split}_\qq$ is the tree-level splitting function associated to~\eqref{Moyal_action}. Explicitly, this splitting function is
\begin{equation}
     \text{Split}_{\qq}(1^+,2^+)= -\kappa \, \cM^{0,3}_\qq (1^+,2^+,-p^-)\times \frac{1}{p^2}=\frac{-\kappa}{2} \frac{[12]^4}{[2p]^2[p1]^2}\frac{[12]_{\qq}}{\la 12\ra} \,,
\end{equation}
where $p=p_1+p_2$ is the momenta in the propagator and we have used the $\overline{\text{MHV}}$ 3-point amplitude \eqref{Moyal_MHV_bar} and the undeformed propagator $1/p^2$. Importantly, because $\sinh(z)$ that appears in~\eqref{Moyal_deformed_square} is an entire function, $[12]_\qq$ introduces no new singularities.  However, in the true collinear limit, $[12]\to0$ as well as $\langle 12\rangle\to0$, so that $[12]_\qq\to[12]$ and we have the same collinear limit as at $\qq=0$.

For celestial holography, what is actually needed is the \emph{holomorphic} collinear limit, where $\langle 12\rangle\to0$ with $[12]$ unchanged. Fortunately, it was argued in~\cite{Ball:2021tmb} that for two positive helicity gravitons, the structure of the collinear limit is exactly the same as that of the true collinear limit. In the Moyal case, this is still true provided we take the limit in a way such that $\langle 1\alpha$ remains non-zero so as to avoid generating singularities from $[12]_\qq$. Parametrizing the momenta in the holomorphic collinear limit as usual by
\begin{equation}
 |p\rangle = \frac{1}{\sqrt{t}}|1\ra = \frac{1}{\sqrt{1-t}}|2\ra \,,
\end{equation} 
the Moyal deformed amplitude $\cM^{1,n}_{\qq}$ must behave in the holomorphic collinear limit $\langle12\rangle\to0$ as
\begin{equation}
\label{holomorphic_collinear_Mq}
    \cM^{1,n}_\qq(1^+,2^+,\dots,n^+) \xrightarrow{1\parallel 2}\frac{-\kappa/2}{t(1-t)}\, \frac{[12]_{\qq}}{\la 12\ra}\  \cM^{1,n-1}_\qq(p^+,3^+,\dots,n^+)\,.
\end{equation}
The crucial difference compared to self-dual gravity is that the overall factor of $[12]$ is deformed to $[12]_\qq$, arising from the fact that the splitting function involves the Moyal deformed vertex between the two collinear particles.

\subsection{$LW_{\wedge}$ algebra from the holomorphic collinear limit}
\label{sec:collinear}

Celestial amplitudes are defined to be the Mellin transform of massless momentum space amplitudes, whose external null momentum $p_i$ can be parametrized by an energy scale $\omega_i$ for each particle and local coordinates $z_i$, $\tilde{z}_i$ on the celestial sphere
\cite{Pasterski:2016qvg}
\begin{equation}
   p_{\mu}=\omega\, z_\alpha \tilde{z}_{\dot{\alpha}} = \frac{\omega}{\sqrt{2}}\left(1+|z|^2,\,z+\tilde{z},\,-\im (z-\tilde{z}),\,1-|z|^2\right)\,,  
\end{equation}
Then a Mellin transform in the energy scale $\omega_i$ gives the celestial amplitude
\begin{equation}
    \widetilde{\mathcal{M}}^{1,n}_\qq(\Delta_i,z_i,\tilde{z}_i)= \left[\prod_{i=1}^n\int_{0}^\infty \frac{d\omega_i}{\omega_i} \,\omega_i^{\Delta_i}\right]\cM_{\qq}^{1,n}(p_i^+).
\end{equation}
And each individual external particle is taken from momentum eigenstate $\phi$ \eqref{plane_wave_normalisation} to boost eigenstate, which we shall label by $G_{\Delta}(z,\tilde{z})$ in the following discussions. We now would like to examine the holomorphic collinear limit of our $\qq$-deformed 1-loop amplitude \eqref{holomorphic_collinear_Mq} in conformal primary basis. After relabeling energy scales of particle 1 and 2 in the collinear regime as $\omega_1=t\omega_p$ and $ \omega_2=(1-t)\omega_p$, performing the $\omega_p$ integral gives
\begin{equation}
\begin{aligned}
  &\widetilde{\cM}^{1,n}_\qq(\Delta_i,z_i,\tilde{z}_i)\ \overset{z_{12}\rightarrow0}{\longrightarrow} \,  -\frac{\kappa}{2}\, \sum_{l=0}^{\infty} \frac{(\qq\, z_{\alpha 1}z_{\alpha 2})^{2l}}{(2l+1)!}\,   \frac{\tilde{z}_{12}^{2l+1}}{z_{12}}\ \times\\
    &\int_0^1  dt\, t^{\Delta_1+2l-2}(1-t)^{\Delta_2+2l-2}\, 
      \widetilde{\cM}^{1,n-1}_\qq(\Delta_1\!+\!\Delta_2\!+\!4l,z_2,\tilde{z}_2\!+\!t\tilde{z}_{12};\cdots)\ +\  \cO(z_{12}^0)\,,
\end{aligned}
\end{equation}
where $z_{ij}=z_i-z_j \, , \tilde{z}_{ij}= \tilde{z}_i-\tilde{z}_j$ and we have used spinor parametrization \eqref{spinor_parametrize} to capture collinear singularity. The additional sum over $l$ comes from expanding the $\sinh$ in $\qq$.
Taylor expanding the celestial amplitude on the right hand side in $\tilde{z}_{12}$ and performing the $t$-integral leads to
\begin{equation}
\begin{aligned}
    &\widetilde{\cM}^{1,n}_\qq(\Delta_i,z_i,\tilde{z}_i)\ \overset{z_{12}\rightarrow0}{\longrightarrow}-\frac{\kappa}{2z_{12}}\sum_{l=0}^{\infty}\frac{(-1)^l(\qq\, z_{\alpha 1}z_{\alpha 2})^{2l}}{(2l+1)!}\ \times\\
    &\left[\sum_{n=0}^{\infty} \frac{
     \tilde{z}_{12}^{n+2l+1}}{n!}\, B(\Delta_1\!-\!1\!+\!2l\!+\!n,\Delta_2\!-\!1\!+2l)\ \tilde{\partial}^n_2\widetilde{\cM}^{1,n\!-\!1}_\qq(\Delta_1\!+\!\Delta_2\!+\!4l,z_2,\tilde{z}_2;\cdots)\right]\\
     &\qquad +\  \mathcal{O}(z_{12}^0)\,,
\end{aligned}
\end{equation}
where $\tilde\partial_2 = \partial/\partial\tilde{z}_2$. From this we can read off momentum space splitting function in boost eigenstate as two conformal primary gravitons approaching each other on the celestial sphere producing the celestial OPE. Here $B(x,y)=\frac{\Gamma(x+y)}{\Gamma(x)\Gamma(y)}$ is the Euler Beta function. Recalling that  conformally soft gravitons are defined by $H^k(z,\tilde{z})=\text{Res}_{\Delta=k}\, G_{\Delta}(z,\tilde{z})$ for $k=2,1,0,-1,\dots$, we obtain the celestial OPE of $H^k$ and $H^j$ as:
\begin{equation}
\label{eq:OPEH}
\begin{aligned}
    &H^k(z_1,\tilde{z}_1)\,H^j(z_2,\tilde{z}_2)\\ 
         &\sim-\frac{\kappa}{2z_{12}}\sum_{l=0}^{\infty}\sum_{n=0}^{1-k-2l}\frac{(-1)^l}{(2l+1)!} \binom{2-k-j-4l-n}{1-2l-j}(\qq z_{\alpha 1}z_{\alpha 2})^{2l}  \frac{\tilde{z}_{12}^{n+2l+1}}{n!} \tilde{\partial}^n H^{k+j+4l}(z_2,\tilde{z}_2).
\end{aligned}
\end{equation}
As a consistency check, we see that taking the $\qq\rightarrow 0$ limit reduces~\eqref{eq:OPEH} to the OPE found in~\cite{Guevara:2021abz, Ball:2021tmb}.\\

In the case of self-dual gravity, one extracts the OPE between conformally soft modes 
\begin{equation}
    H^k_n(z)=\oint_{|\tilde{z}|<\epsilon} \frac{d\tilde{z}}{2\pi i}\,\tilde{z}^{n+\frac{k-4}{2}}\,H^k(z,\tilde{z})
\end{equation} 
as a contour integral of the $\qq=0$ limit of~\eqref{eq:OPEH}. To obtain the $Lw_{\wedge}$-algebra, one must then use Vandermonde identity $\sum_{t=0}^r \binom{n}{t}\binom{m}{r-t}=\binom{m+n}{r}\,$ to perform the residual sum \cite{Ball:2021tmb}. We can perform the same contour integrals and residual sum in the $\qq$-deformed theory where we need the following generalised Vandermonde identity
\begin{equation} 
      \sum_{t=0}^r\  [t]_l\,\binom{n}{t}\binom{m}{r-t}=\frac{[r]_l\,[m]_l}{[m+n]_l}\binom{m+n}{r}\,,
\end{equation}
which will be proved in section~\ref{sec:Pochhammer}.

In order to perform the contour integrals over equation \eqref{eq:OPEH} in $\tilde{z}_1$ and $\tilde{z}_2$, we use
\begin{equation}
\label{eq:contouridentity}
    \oint_{|\tilde{z}_1|<\epsilon} d\tilde{z}_1 \frac{\tilde{z}_1^{m+\frac{k-4}{2}}}{2\pi i} \tilde{z}_{12}^{n+2l+1}=\frac{(n+2l+1)!}{(\frac{2-k}{2}-m)!\,(n+2l+m+\frac{k}{2})!}(-\tilde{z}_2)^{m+n+2l+\frac{k}{2}}
 \end{equation} 
if $-\frac{k}{2}-m\leq n+2l$, and zero otherwise.
This shows that
\begin{align}
    &H^k_m(z_1)\,H^j_{n}(z_2) \nonumber\\ 
    \quad&\sim  -\frac{\kappa}{2z_{12}} \sum_{l=0}^{C(m,n,k,j)} \sum_{t=-k/2-m-2l}^{1-k-2l} \left\{\frac{(-1)^{l+m+t+2l+k/2}}{(2l+1)!\,t!} (qz_{\alpha 1}z_{\alpha 2})^{2l}\binom{2-k-j-4l-t}{1-2l-k-t}\right. \nonumber \\ 
    \quad&\qquad\qquad\times\  \left.\frac{(t+2l+1)!}{(\frac{2-k}{2}-m)!\,(t+2l+m+\frac{k}{2})!}  \left[\frac{2-k-j}{2}-2l-m-n\right]_t H^{k+j+4l}_{m+n}(z_2)\right\}
\end{align}
where $C(m,n,k,j)=\lfloor\tfrac{1}{4}(|m+n|+\frac{2-k-j}{2}) \rfloor$ is the upper bound for $l$, for which the inequality in the second case of equation \eqref{eq:contouridentity} holds. Note that when $l=0$, this reduces to the undeformed self-dual gravity calculation where the sum over $t$ can be performed straight away. It is more subtle here and requires the use of two identities in appendix~\ref{sec:Pochhammer}. More precisely, after shifting $t\rightarrow t-k/2-m-2l$ and rearranging the terms to collect terms that involves $t$, the right hand side becomes
\begin{align}
\label{eq:lasttsum}
 &H^k_m(z_1)H^j_{n}(z_2)\sim -\frac{\kappa}{2z_{12}}\sum_{l=0}^{C(m,n,k,j)}\frac{(-1)^{2l+1-k/2+m}}{(2l+1)!}\\ \nonumber
 &\frac{((2-k)/2-m+(2-j)/2-n-2l-1)!}{(\frac{2-k}{2}-m)!\,(\frac{2-j}{2}-n)!} (qz_{\alpha 1}z_{\alpha 2})^{2l} F(k,j,l,m,n)\,H^{k+j+4l}_{m+n}(z_2)\,.
\end{align}
where $F(k,j,l,m,n)=\sum_{t=0}^{(2-k)/2+m} 
[t-k/2-m+1]_{2l+1}\binom{j+2l-2}{1-k/2+m-t}\binom{\frac{2-j}{2}-n}{t} $ is the sum over $t$, the first term $[t-k/2-m+1]_{2l+1}$ can be expanded using lemma \ref{lem:Pochhammer1}, which gives 
\begin{equation}
    \left[t+\frac{2-k}{2}-m\right]_{2l+1}=\sum_{i=0}^{2l+1}\binom{2l+1}{i}[t]_{k}\left[\frac{2-k}{2}-m\right]_{2l+1-i}\,,
\end{equation}
Then, the generalized Vandermonde-identity from lemma \ref{lem:Pochhammer2} can be used to perform the sum over $t$ which leads to 
\begin{align}
F(k,j,l,m,n)& =(-1)^{2l+(2-k)/2+m} \\
   & \sum_{i=0}^{2l+1}\binom{2l+1}{i}\left[\frac{2-k}{2}-m\right]_{2l+1-i}\frac{[\frac{2-k}{2}+m]_i[\frac{2-j}{2}-n]_i}{[\frac{j-2}{2}-n+2l]_i}\binom{\frac{j-2}{2}-n+2l}{\frac{2-k}{2}+m}\,.\nonumber
\end{align}   
Absorbing the $(-1)^{2l+(2-k)/2+m}$ into the remaining binomial coefficients using $(-1)^k \binom{n}{k}=\binom{k-n-1}{k}$ leaves us with
\begin{align}
    F(k,&j,l,m,n)=\frac{(\frac{2-k}{2}+m+\frac{2-j}{2}+n-2l-1)!}{(\frac{2-k}{2}+m)!\,(\frac{2-j}{2}+n)!}\\\nonumber
    & \sum_{i=0}^{2l+1}\binom{2l+1}{i}\left[\frac{2-k}{2}-m\right]_{2l+1-i}\left[\frac{2-k}{2}+m\right]_i\left[\frac{2-j}{2}-n\right]_i\left[\frac{2-j}{2}+n\right]_{2l+1-i}.
\end{align}

\noindent After substituting this back into equation \eqref{eq:lasttsum} one finds the desired result that
the OPE of soft graviton modes in the $\qq$-deformed theory is
\begin{align}
    &H^k_m(z_1)\,H^j_{n}(z_2)\sim \\
    \nonumber&\frac{-\kappa}{2z_{12}}\sum_{l=0}^{C(k,j,m,n)}\frac{(-1)^l(\qq\, z_{\alpha 1}z_{\alpha 2})^{2l}}{(2l+1)!} \bigg(\frac{(\frac{2-k}{2}\!-\!m\!+\!\frac{2-j}{2}\!-\!n\!-\!2l\!-\!1)!\,(\frac{2-k}{2}\!+\!m\!+\!\frac{2-j}{2}\!+\!n\!-\!2l\!-\!1)!}{(\frac{2-k}{2}-m)!(\frac{2-j}{2}-n)!(\frac{2-k}{2}+m)!(\frac{2-j}{2}+n)!}\bigg)\times\\
    \nonumber&\ \sum_{i=0}^{2l+1}(-1)^i\binom{2l\!+\!1}{i}\left[\frac{2\!-\!k}{2}\!-\!m\right]_{2l+1-i}\left[\frac{2\!-\!k}{2}\!+\!m\right]_i\left[\frac{2\!-\!j}{2}\!-\!n\right]_i\left[\frac{2\!-\!j}{2}\!+\!n\right]_{2l+1-i}H^{k+j+4l}_{m+n}(z_2)\,,
\end{align}
where the upper limit of the sum 
\[
C(k,j,m,n)=\left\lfloor\frac{|m+n|}{4}+\frac{2-k-j}{8} \right\rfloor\, .
\]
Using the relabelling $k=4-2p$, we rewrite the modes to absorb some of the factorials in the coefficient 
\begin{equation}
W^p_m(z)=\frac{1}{\kappa} (p-m-1)!\,(p+m-1)!\,H^{4-2p}_m(z)
\end{equation}
as discussed in section \ref{sec:Chalmers-Siegel} to obtain
\begin{equation}
\begin{aligned}\label{OPE_result}
    & W^p_m(z_1)\,W^q_n(z_2)\sim \frac{1}{2z_{12}} \sum_{l=0}^{C(p,q,m,n)} \frac{(-1)^l (\qq\,z_{\alpha 1}z_{\alpha 2})^{2l}}{(2l+1)!}\sum_{i=0}^{2l+1}(-1)^i\binom{2l+1}{i}\\
    &[p-1+m]_{2l+1-i}[p-1-m]_i[q-1+n]_i[q-1-n]_{2l+1-i}W_{m+n}^{p+q-2-2l}(z_2)
\end{aligned}
\end{equation}
with $C(p,q,m,n)=\lfloor\tfrac{1}{4}(|m+n|+p+q-3) \rfloor$, which enforces the wedge condition.
Comparing \eqref{OPE_result} and \eqref{W_algebra}, we see that we have arrived at $W_\wedge$ on the nose. Remember the fact that although the sum over $l$ does not have a upper limit in \eqref{W_algebra}, an automatic cut off is in place for the wedge condition.

\section{Discussion}
\label{sec:Discussion}

In this paper, we have studied the Moyal deformed self-dual gravity theory~\eqref{Moyal_action}. We have examined the form of the collinear splitting function of their $n$-pt all $+$ helicity amplitudes and found that the only difference with the usual amplitude involves the splitting function deformed by $[ij]\to[ij]_\qq$. After Mellin transforming this collinear limit to extract the contribution of conformally soft gravitons, we found that the deformed splitting function gives rise to $LW_\wedge$. This is the loop algebra of $W_\wedge$, which is itself an augmentation of the wedge subalgebra of symplecton algebra to allow half-integer values for $p,m$. We note that the symplecton algebra is the unique $W(\mu)$-algebra which allows this augmentation. Thus $W_\wedge$ is the unique deformation of $w_\wedge$ as a Lie algebra. Just like $Lw_\wedge$ in self-dual gravity, the $LW_\wedge$ algebra is a perturbatively exact symmetry of the Moyal deformed self-dual gravity. Thus the results of this paper suggest that lifting celestial holography from $w_{1+\infty}$ to its deformation is more likely related to turning on non-commutativity in the bulk, rather than quantum corrections. 

\medskip

There are many possible avenues worthy of further explanation. Firstly, it would be good to know the form of the all plus 1-loop amplitudes in the Moyal deformed theory. Any Feynman diagram contributing to these amplitudes would involve exactly $n$ Moyal vertices. One possible conjecture for the amplitude that seems to naturally generalize the all plus amplitudes in self-dual gravity is
\begin{equation}\label{Moyal_deformed_Bern}
    \cM_\qq^{1,n}(p_i^+) = \sum_{1\leq a<b\leq n}\sum_{M} \ h_\qq(a,M,b)\,h_\qq(b,\overline{M},a)\,(\tr_{\qq}(aMb\overline{M}))^2\,\tr(aMb\overline{M})\,,
\end{equation}
where $h_\qq(a,M,b)$ is the $\qq$-deformed half-soft function
\begin{equation}\label{Moyal_deformed_h}
    h_\qq(a,M,b) = {\det}'(H^M_\qq) = \frac{\det({}^r_c H^M_\qq)}{\la ar\ra\la br\ra\,\la ac\ra\la bc\ra} 
\end{equation} 
where the $\qq$-deformed Hodges type matrix $H_\qq^M$ is defined  as 
\begin{equation}
\label{modified_Hodges}
(H_\qq^M)_{ij}\ =\ \begin{cases}
\displaystyle{\ \frac{[ij]_\qq}{\la ij\ra}}
\qquad &\text{for $i\neq j$,}\\ 
 \displaystyle{\ -\sum_{k\in M\setminus\{i\}} \frac{[ik]_\qq}{\la ik\ra}\frac{\la ak\ra\la bk\ra}{\la ai\ra\la bi\ra}}
 \qquad\qquad&\text{when $i=j$}
\end{cases}
\end{equation}
using the $\qq$-bracket defined in~\eqref{Moyal_deformed_square}. 
Similarly, $\tr_{\qq}(aMb\overline{M})$ is the $\qq$-deformed trace function
\begin{equation}
    \tr_\qq(aMb\overline{M}) = \la a|K_{M}|b]_{\qq}\la b|K_{\overline{M}}|a]_{\qq}+\la b|K_{M}|a]_\qq \la a|K_{\overline{M}}|b]_\qq\,.
\end{equation}
The fact that the half-soft function of self-dual gravity can be written as  this Hodges matrix at $\qq=0$ was found in~\cite{Feng:2012sy}, and it may be possible to find a recursive proof of~\eqref{Moyal_deformed_Bern} using the methods of that paper.

\medskip

While the deformed theory~\eqref{Moyal_action} involves no new fields beyond the usual graviton (at least at the level of the effective field theory we have considered),  alternative approach to constructing gravitational theories possessing $LW_\wedge$ symmetry have been pursued by~\cite{Mago:2021wje,Ren:2022sws}, who consider certain non-minimally scalars in addition to gravity in the bulk, and by~\cite{Krasnov:2021nsq,Tran:2022tft} who consider a higher spin theory containing particles of all integer spin. It would be interesting to understand the relations between these theories and ours.

\medskip

It is also possible to couple self-dual Yang-Mills to Moyal deformed self-dual gravity. One way to achieve this is to couple the gravitational action~\eqref{Moyal_action} to the action 
\begin{equation}
    \label{Moyal-sdg-sdYM}
S[\phi;\bar\chi,\chi] = \int \tr\left(\tilde\chi\Box\chi  + \kappa\,\tilde\chi \left\{\p^{\dal}\phi,\p_{\dal}\chi \right\}_\qq + \frac{g}{2}\tilde\chi\left[\p^{\dal}\chi,
\p_{\dal}\chi\right]_\star\right)  \, \d^4x\,,
\end{equation}
where $\chi$ and $\tilde\chi$ represent the positive and negative helicity modes of the gluon. The Moyal bracket appears in the coupling $\tr(\tilde\chi \{\p^{\dal}\phi,\p_{\dal}\chi\}_\qq) = \tr(\tilde\chi\, \p^{\dal}\p^{\dot\beta}\phi\,\p_{\dal}\p_{\dot\beta}\chi) + \cO(\qq^2)$ between the positive helicity graviton $\phi$ and the gluons, while it also seems natural to deform the self-dual Yang-Mills interaction $\tr(\tilde\chi\,[\p^{\dal}\chi,\p_{\dal}\chi]) = f^a_{bc}\,\tilde\chi_a\,\p^{\dal}\chi^b\,\p_{\dal}\chi^c$ to the non-commutative  $f^a_{bc}\,\tilde\chi_a(\p^{\dal}\chi^b\star\p_{\dal}\chi^c)=f^a_{bc}\,\tilde\chi_a\{\p^{\dal}\chi^b\star\p_{\dal}\chi^c\}_{\qq}$ constructed using the Moyal star. We expect that in this coupled  theory, 
the gauge and mixed gauge-gravity holographic symmetry algebras considered in \cite{Guevara:2021abz} also get deformed. In particular, we would presume that the mixed gluon-graviton amplitudes in this theory show that the conformally soft gluon generators transform as an $LW_\wedge$ module.

\medskip

Finally,~\cite{Ooguri:1991fp} proposed that the string field theory of the $\cN=2$ string is governed by an action that can be written as
\begin{equation}
    \label{Ooguri-Vafa_action}
     \frac{1}{\kappa^6}\int \frac{1}{2}\partial^{\dot{\alpha}}\varphi\,\bar{\partial}_{\dot{\alpha}}\varphi + \frac{1}{3!} \,\varphi\, \{\bar{\p}^{\dal}\varphi,\bar{\p}_{\dal}\varphi\}\ \d^4x\,,
\end{equation} 
where $\bar{\p}_{\dal}= \hat{\alpha}^\alpha\p_{\dal\al}$. Unlike the Chalmers-Siegel action~\eqref{Chalmers-Siegel}, this action contains only a single, real  scalar $\varphi$ of mass dimension $-2$,
 interpreted as a deformation of the K{\"a}hler potential $K= \tilde{z}^{\dal} z_{\dal} + \varphi$ around flat $\mathbb{R}^{2,2}$. The equation of motion is the first Plebanski equation, ensuring that the metric associated to the pseudo-K{\"a}hler metric is Ricci flat. One can again deform this theory by replacing the Poisson bracket by a Moyal bracket and it would be interesting to understand how this arises from the worldsheet perspective. The most obvious possibility is to turn on a background $B$-field~\cite{Lechtenfeld:2000nm}. However, it is curious to note that $\qq$-expansion of the resulting 3-pt tree amplitude has the same structure as corrections to the 3-pt amplitude coming from higher genus\footnote{Since~\eqref{Ooguri-Vafa_action} preserves only a U(1,1) subgroup of SO(2,2) and, unlike~\eqref{Chalmers-Siegel}, full Lorentz invariance cannot be restored by any choice of normalization for $\varphi$. Thus one cannot argue that the 3-pt amplitudes in this theory are completely fixed by Lorentz invariance and little group scaling.} worldsheets found by~\cite{Bonini:1990bf,Lechtenfeld:1999gd, Marcus:1992wi, Junemann:1999xj}.

\acknowledgments
We are grateful for Tim Adamo, Nathan Berkovits, Prahar Mitra, Atul Sharma, Peter Wildemann and especially Roland Bittleston for useful discussions. WB is supported by the Royal Society studentship. SH is partly supported by St. John's College, Cambridge. The work of SH \& DS has been partly funded by STFC HEP Theory Consolidated grant ST/T000694/1.

\appendix
\section{Two identities involving Pochhammer symbols} 
\label{sec:Pochhammer}
Throughout the paper, we use a notation in which the falling Pochhammer symbols are given by $[x]_{n}=\Gamma(x+1)/\Gamma(x-n)$.
\begin{lemma}
\label{lem:Pochhammer1}
Let $x,y \in \mathbb{C}\setminus \mathbb{Z}_{<0}$ and $n\in \mathbb{N}$, then
\begin{align}
[x+y]_{n}=\sum_{k=0}^n\binom{n}{k}[x]_k[y]_{n-k}\,,
\end{align}
\end{lemma}
\begin{proof}
The base case $n=0$ is trivially true and the inductive step is similar to the proof of the binomial expansion
\begin{equation}
\begin{aligned}
 [x+y]_{n+1}&= (x+y-n)[x+y]_{n}\\
 &= (x-k)\sum_{k=0}^n\binom{n}{k}[x]_k[y]_{n-k}+ (y-n+k)\sum_{k=0}^n\binom{n}{k}[x]_k[y]_{n-k}\\
 &= \sum_{k=0}^n\binom{n}{k}[x]_{k+1}[y]_{n-k}+\sum_{k=0}^n\binom{n}{k}[x]_k[y]_{n+1-k}\\
 &= \sum_{k=0}^{n+1} \binom{n+1}{k} [x]_k[y]_{n+1-k}\,.
\end{aligned}
\end{equation}

\end{proof}
\begin{lemma}[Generalized Chu-Vandermonde identity]
\label{lem:Pochhammer2}
Let $m,n\in \mathbb{C}\setminus \mathbb{Z}_{<0}$ and $r,l\in \mathbb{N}$, then
\begin{align}
      \sum_{t=0}^r [t]_l\binom{n}{t}\binom{m}{r-t}=\frac{[r]_l[m]_l}{[m+n]_l}\binom{m+n}{r}\,,
\end{align}
\end{lemma}
\begin{proof}
For fixed $l$, we shall start with the L.H.S. and show that the generating function in $r$ agree with the R.H.S. 
\begin{equation}
\begin{aligned}
    \sum_{r=0}^{m+n} \Bigg(\sum_{t=0}^r [t]_l\binom{n}{t}\binom{m}{r-t} \Bigg)x^{r-l}&= \left(  \sum_{r_1=0}^{n}[r_1]_l\binom{n}{r_1}x^{r_1-l}\right) \left( \sum_{r_2=0}^m \binom{m}{r_2}x^{r_2}\right)\\ 
    &=\left(\big(\frac{d}{dx}\big)^l (1+x)^n\right)(1+x)^m\\ &=[n]_l (1+x)^{n+m-l}\\
    &=[n]_l\sum_{r=0}^{m+n-l}\binom{m+n-l}{r}x^r\\
     &=[n]_l\sum_{r=l}^{m+n}\binom{m+n-l}{r-l}x^{r-l}\\
     &=\sum_{r=0}^{m+n}\left(\frac{[n]_l\,[r]_l}{[m+n]_l}\binom{m+n}{r}\right)x^{r-l},
\end{aligned}
\end{equation}
where we have shifted the index $r\rightarrow r+l$ and used the following identity for binomial coefficients
\begin{align}
    \binom{a_1-l}{a_2-l}=\frac{[a_2]_l}{[a_1]_l}\binom{a_1}{a_2} \,.
\end{align}
\end{proof}

\newpage
\bibliographystyle{JHEP}
\bibliography{ref}

\end{document}